\documentclass[11pt]{article}

\usepackage{amssymb,amsmath, amsthm, setspace, color, graphics}
\usepackage{dsfont}              
\usepackage[dvips]{graphicx}
\usepackage[english]{babel}
\usepackage{float}
\usepackage{caption}
\usepackage{subcaption}
\usepackage[authoryear]{natbib}
\usepackage[pdfborder={0 0 0},final=true,colorlinks=true,linkcolor=magenta,citecolor=blue]{hyperref}
\numberwithin{equation}{section}
\theoremstyle{plain}
\newtheorem{thm}{Theorem}[section]

\newtheorem{lem}[thm]{Lemma}
\newtheorem{prop}[thm]{Proposition}
\theoremstyle{plain}
\newtheorem{assumption}[thm]{Assumption}      
\theoremstyle{definition}

\theoremstyle{definition}
\newtheorem{defn}{Definition}[section]
\theoremstyle{remark}
\newtheorem{rem}{Remark}[section]
\renewcommand{\thefootnote}{\fnsymbol{footnote}}

\begin{document}
\title{\vspace{-8ex}\textbf{Bankruptcy Risk Induced by Career Concerns of Regulators }\footnote{This paper is based on \S 2.4 and \S 2.5 of the paper entitled ``A Regulator's Exercise of Career Option To Quit and Join A Regulated Firm's Management with Applications to Financial Institutions". We are indebted to Samuel Myers, Jr. and Gary Dymski for detailed comments, and Haim Abraham for helpful conservations on earlier drafts.  We thank conference participants at the \emph{Academy for Behavioral Finance \& Economics Meeting 2012}, Robert B. Durand (discussant),  \emph{Southern Finance Association Meeting 2012}, Rachel Graefe-Anderson (discussant), Paul Kupiec, Leonard Nakamura, Doug Robertson, Ernesto Dal B\'{o}, Oliver Martin, Linwood Tauheed, Michael Crew and Rebel Cole for their comments and advice, and Sam Cadogan for his research assistance. We thank Ramazan Gen\c{c}ay (Editor), an anonymous co-editor, and two anonymous referees for helpful comments and suggestions-especially the request to include a solution to the regulator's program which generated \autoref{subsec:RegulatorProgram}. Cadogan is indebted to Jan Kmenta for the idea of embedding behavioural parameters in decision variables. Any errors which may remain are ours.  He  gratefully acknowledges the research support of the Institute for Innovation and Technology Management (IITM); and Research Unit in Behavioural Economics and Neuroeconomics (RUBEN).}}
\author{
    John A. Cole
    \thanks{Corresponding author: Department of Economics and Finance, School of Business \& Economics, North Carolina A\& T State University, 1621 East Market Street, Greensboro, NC 27411; e-mail: \textcolor[rgb]{0.00,0.00,1.00}{\href{mailto:: jacole@ncat.edu}{jacole@ncat.edu}}; Tel: (336) 334-7744 Fax: (336) 256-2055}
    \qquad Godfrey Cadogan
    \thanks{University of Cape Town, School of Economics, Research Unit in Behavioural Economics and Neuroeconomics, Private Bag 5-3b, Rondebosch 7701; and Institute for Innovation and Technology Management, Ted Rogers School of Management, Ryerson University, 575 Bay, Toronto, ON M5G 2C5  e-mail: \textcolor[rgb]{0.00,0.00,1.00}{\href{mailto:gocadog@gmail.com}{gocadog@gmail.com}}; Tel: (786) 329-5469.}~\\
    }
\date{\vspace{-1ex}\today\vspace{-4ex}} 
\maketitle
\thispagestyle{empty}
\begin{abstract}
   \noindent We introduce a model in which a regulator employs mechanism design to embed her human capital beta signal(s) in a firm's capital structure, in order to enhance the value of her post career change indexed executive stock option contract with the firm. We prove that the agency cost of this revolving door behavior increases the firm's financial leverage, bankruptcy risk, and affects estimation of firm value at risk (VaR).
\end{abstract}
\newpage
\renewcommand{\thefootnote}{\arabic{footnote}}\setcounter{footnote}{0} 
\tableofcontents
\thispagestyle{empty}
\newpage
\setcounter{page}{1}
\pagenumbering{arabic}
\section{Introduction}\label{sec:Introduction}
This paper extends the literature on contingent claim analysis of career option, introduced by \citet[Sec.~II]{BodieRuffinoTreussard2008}, to firm bankruptcy risk induced by career concerns\footnote{See e.g., \citet{Holmstrom1999}} of a regulator in a revolving door setting\footnote{See e.g., \citet{Che1995} and \citet{Salant1995}}. The intuition behind our model is regulators have incentives to regulate firms that can provide them with more lucrative career change employment.  They bring value to the firm either by reducing shareholders’ financing costs, or by withholding expected consumer surplus that might arise in regulatory pricing decisions\footnote{\citet{HoustonJiangChen2013}; \citet{LitzenbergerSosinL1979}.}.  In either case, the regulator's employment option value depends on the firm’s levered beta coefficient.  Relative to any given beta level that might hold in the case of a disinterested regulator, we argue there will be a beneficial impact on the firm’s beta when a regulator exercises her career option to join the firm\footnote{See e.g., \cite{AcemogluJohnsonKermaniKwakMitton2013}.}. Whereas \citet{Tressaud2007} and \citet{BodieRuffinoTreussard2008} solve the problem of \emph{when} to switch careers, our model's focus is on \emph{what} happens once a regulator decides to switch career and take an executive level position with a firm she regulates. Here, she embeds her human capital beta\footnote{This is a beta pricing result in \cite[eq(22)]{JagannathanWang1996} where return on human capital is used to augment a conditional CAPM model. In \cite[pg.~43]{BodieTreussard2007} human capital beta refers to a lifetime portfolio allocation fraction that does not change.} in the firm's capital structure and observes a valuation effect. As a part of the decision process to switch career choices she negotiates an indexed option contract that includes firm value considerations\footnote{\citet{DittmannMaugSpalt2013} highlight disincentive effects of indexed contracts.}.

A seminal paper by \citet[pg.~429,~eq~(6)]{SpiegelSpulber1994} shows how regulators' welfare functions are tied to a firm's capital structure through debt. And a more recent paper by \citet{BortolottiCambiniRondiSpiegel2011} report that highly levered firms receive more favorable regulation. In our model, the regulator utilizes her human capital beta as a signalling device. High betas signal tough regulator and regulator expertise to prospective employers. Low betas signal lax regulator who curries favor, and lobbying potential to consulting firms. See e.g., \citet{DehaanKediaKohRajgopal2012}. We show how the beta parameter can be used to manipulate the barrier in an indexed executive stock option contract in order to increase the value of the regulator's post career-change early exercise premium. See e.g., \citet{GoldmanSlezak2006}. In that context, the volatility of firm value increases. See e.g., \citet{JohnsonTian2000a}; \citet{Jorgensen2002}. Thus, we employ a signal dependent option pricing model by \citet{Guo2001} and \citet{Shepp2002} to identify firm bankruptcy risk in bad states by comparing firm vega \footnote{See e.g., \citet[pp.~359-360]{Hull2006}}, and firm value at risk (VaR), with and without regulator signals embedded in capital structure. In a broad sense, our model falls under rubric of the \citet[pg.~333]{JensenMeckling1976} agency theory of capital structure, as opposed to the \citet{ModiglianiMiller1958} capital structure irrelevancy model.
\subsection{Significance of results}\label{subsec:SigbifResults}
Our results bear significance for financial policy. For example, they are consistent with \citet[pg.~164]{Kane2009} who proposed a model with asymmetric information in which opacity grows with a financial firm's product line. There, opacity causes regulators to overestimate firm capital and underestimate firm leverage and volatility. Even though \citet{Kane2009} deals with regulatory failure, the salient characteristics of his model applies to ours. For our firm overestimates the contractual benchmark, and underestimates its volatility because of opacity in the regulator's ability to manipulate the benchmark with human capital beta. Given the increased firm bankruptcy risk predicted by that behaviour in our model, \citet[pg.~155]{Kane2010} advocacy of an incentive pay scheme for regulators is \emph{apropos}. There, bonuses or raises are framed as deferred compensation which would be forfeited if a crisis occurs within a given time frame after the regulator resigned. This is tantamount to a monitoring device in the form of a performance vested barrier option with a down-and-out feature. See e.g., \citet[pg.~8]{JohnsonTian2000b}. In effect, it mitigates the regulator's rent seeking via manipulation of her human capital beta signals.

Our model also applies to ``regulation by assimilation" instruments  such as a firm's value-at-risk (VaR). For instance,  VaR was developed by private industry as a risk management tool and subsequently ``assimilated" or adopted by regulators. \citet[pg.~629]{Macey2013}. Our model shows how estimates of the VaR statistic is inflated by regulator signals. Given that regulators permit firms' to exercise discretion in computing their VaR statistic, our formulae suggests consequential moral hazard from firms reporting the lower of two numbers derived from (1) regulator suggested formula; and (2) firm internal formula.

The rest of the paper proceeds as follows: \autoref{sec:StockOptionsWithRegSignals} develops the abstract mechanism design, and abstract call option pricing strategy for a regulator who wishes to quit and join firm management at a future date known only to the regulator. In \autoref{sec:RegulatorOptionCapitalStructure} we apply the mechanism and option pricing strategy by and through (i) the regulator’s human capital beta; (ii) the stock price benchmark which is a decision criterion; and (iii) the process whereby the regulator embeds her human capital beta in the firm’s capital structure, and (iv) we solve the regulator’s option problem. In \autoref{sec:WarningSignalsVegaAvgLev} we present the main result of the paper's bankruptcy risk criterion, and conclude in \autoref{sec:Conclusion}.

\section[Regulator career call option with regulatory signalling]{Regulator career call option contract with regulatory signalling}\label{sec:StockOptionsWithRegSignals}
\subsection{Regulator mechanism design}\label{subsec:RegulatorMechanismDesign}
In this section, we briefly describe the abstract mechanism design\footnote{The interested reader is referred to \citet{BaligaSjostrum2008} for a brief survey or \citet{Borgers2008} for a more detailed exposition on mechanism design theory, and \citet[pg.~208,~Appendix]{DalBo2006} for details on application and implementation in the context of a revolving door model. } utilized to embed regulator signals in the firms capital structure to affect the value of her [call] option contract upon career change. An important paper by \citet{Che1995} analyzes revolving door behaviour in an hierarchical structure consistent with mechanism design. There, the principal (the government) hires an agent (the regulator) to monitor the behaviour of another agent (the firm) to effect compliance with given legislation. By contrast, \citet{Salant1995} considers revolving door behaviour in the sub-game between the regulator and the firm. Our model is distinguished because we consider a mechanism design in which a regulator employs (1) costly human capital beta signals to secure (2) lucrative indexed ESOP contracts upon termination of its contract with the principal. Our regulator has sufficient lead time to observe and compare her current compensation package to prospective managerial compensation with ESOP she receives upon exercising a career option to join firm management. Thus, her risky compensation is based on a prospective equity stake in the regulated firm. Related papers by \citet{RuffinoTressaurd2007} and \citet[pp.~4-5]{Tressaud2007} used a utility model approach to American style option\footnote{See e.g., \citet{Myneni1992} for a review of the literature.} to analyze career choices. However, \citet{Carpenter1998} showed that extension of an American style option model works just as well as elaborate utility maximizing models. Thus, we pose and solve the regulator's problem in the context of an American style indexed option without resorting to utility theory.

We implement our model as follows. Let $\Theta$ be a \emph{type} space for regulators, i.e., tough or lax, and $\Omega$ be a sample space for states of nature, $\mathcal{F}_t$ be the $\sigma$-field of Borel measurable subsets of $\Omega$ at time $t$, $P$ be a probability measure on $\Omega$, and $\mathds{F}$ be a filtration of information $\mathcal{\{F\}}_{t\ge 0}$. Thus, the probability space in our model is characterized by $(\Omega,\mathcal{\{F\}}_t,\mathds{F},P)$. Let $X$ be an outcome space, and $f$ be a stochastic choice function such that $f:\Omega\times\Theta\rightarrow X$ is a direct mechanism. Based on her type, the regulator sends a state dependent \emph{signal} $s$ into some strategy space $G$, i.e. $s:\Omega\times\Theta\rightarrow G$ to affect the firm's response\footnote{For example, a regulator could provide a more favorable report on a firm than warranted by the facts. See e.g., \href{http://dealbook.nytimes.com/2013/10/10/bank-examiner-was-told-to-back-off-goldman-suit-says/?hp\&\_r=0}{http://dealbook.nytimes.com/2013/10/10/bank-examiner-was-told-to-back-off-goldman-suit-says/?hp\&\_r=0}. Last visited 2013/12/01}. On the basis of interactions in $G$, i.e. the firms response, the regulator maps a payoff function or contract $C$ into $X$, i.e $C:G\rightarrow X$. For instance, the space $G$ may comprise the signal $s(\omega,\theta),\;\theta\in\Theta,\;\omega\in\Omega$ obtained from the regulator's welfare maximization\footnote{For example, in a deterministic setting with regulator type $\theta$, let $\Pi(\cdot,\theta)$ be firm profit, $CS(\cdot,\theta)$ be consumer welfare, $W$ be the regulators welfare function, and $\alpha$ be a decision weight such that $W(\cdot,\theta;\alpha)=(1-\alpha)CS(\cdot,\theta) + \alpha\Pi(\cdot,\theta),\;\;0\leq\alpha\leq 1$. Solve an optimization problem for $W$ to obtain an expression for a control variable such as price $p(\theta;\alpha)$ that affects firm profit $\Pi$. Here, strategy space $G$ will be characterized by $p(\theta;\alpha)$, and the ``strategic form" taken by $p$. If $\alpha > \tfrac{1}{2}$, then greater decision weight is placed on firm profit which may be a function of capital structure. See e.g., \citet[pg.~429]{SpiegelSpulber1994}.}; a menu of profits $\Pi$ for the firm, and the strategic form $g$. So that
\begin{equation}
   G= \{s(\omega,\theta),\Pi,g\}
\end{equation}
The regulator builds personal value in the space, given their perceived capacity to affect regulations’ interpretation in ways that benefit the prospective employer firm. ”Undeniably, regulation controls the strategy space available to the firm in this case. However, the regulator's signal $s(\omega,\theta)$ is subject to the vagaries of the state of nature $\omega$ and her type $\theta$..
\begin{defn}[(PC): Regulator's participation constraint]\label{defn:RegParticipateConstraint}~\\
   Let $K$ be the compensation package of the regulator, assumed fixed over the epoch of career mobility, and $V(\omega,\theta,\Pi)$ be the signal dependent value of the firm as a function of profits. The regulator will participate in exercise of a state contingent career [call] option contract when
   $$C(s(\omega,\theta))=(\alpha V(\omega,\theta,\Pi)-K)^+$$
   for some equity stake $0<\alpha<1$ \hfill $\Box$
\end{defn}
\begin{lem}[(IC): Regulator's incentive compatibility constraint]\label{defn:RegIncentiveConstraint}~\\
   Given the probability measure space $(\Omega,\mathcal{\{F\}}_t,\mathds{F},P)$, let $\theta$ be the regulator's true type and $\widehat{\theta}$ be any other reported type. Then for some [Borel] measureable event $A\in\mathcal{F}_t$ we have
   $$\int_A C(s(\omega,\theta))dP(\omega)\geq\int_A C(s(\omega,\widehat{\theta}))dP(\omega)$$
   That is, the regulator has incentive to reveal her true type $\theta$ because the expected value of her career option is greater than it would be otherwise if she does not. \hfill $\Box$
\end{lem}
\begin{proof}
    See e.g., \autoref{apx:ProofOfICconstraint}\footnote{We thank Haim Abraham for requesting a proof of this constraint.}.
\end{proof}
\begin{defn}[Revelation principle]\label{defn:RegRevelationPrinciple}~\\
Given any feasible career option mechanism $C:G\rightarrow X$, there exist a feasible direct mechanism $f:\Omega\times\Theta\rightarrow X$ such that for $A\in\mathcal{F}$ we have the isomorphism $\cong$ given by
\begin{align}
    f(\omega,\theta)&\cong (C\circ s)(\omega,\theta)\label{eq:RegulateTopologyLifting}\\
    \int_A f(\omega,\theta)dP(\omega)&=\int_A(C\circ s)(\omega,\theta)dP(\omega)
\end{align}
\hfill $\Box$
\end{defn}
\begin{rem}
   $f$ is the contract that would be observed with perfect information and full disclosure. However, it may violate conflict of interest or anti-corruption laws. See e.g., \citet{Zeume2012}. By contrast, the [indirect] contract $(C\circ s)$ contains elements known only to the regulator, and they may be revealed over time. \hfill $\Box$
\end{rem}
On the basis of the foregoing definitions we formulate the following
\begin{prop}[Regulator's abstract career call option contract with regulated firms]\label{prop:RegulatorAbstractOption}
   Given $(\Omega,\mathcal{F}_t,\mathds{F},P)$, a type space $\Theta$ for regulators, $G$ a strategy space, $C$ a contract, $\mathcal{M}=(G,C)$ an abstract mechanism, and $s$ be a regulatory signalling function such that $s:\Omega\times\Theta\rightarrow G$; let  $f:\Omega\times\Theta\rightarrow X$ be the regulator's stochastic choice function mapped into outcome space $X$. Specifically, let $C$ be the regulator's career [call] option contract with regulated firms. We claim that the abstract value of the option is given by the mechanism: $$f(\omega,\theta)\cong (C\circ s)(\omega,\theta)$$
   where $\cong$ represents an isomorphism. \hfill $\Box$
\end{prop}

\begin{rem}
  The isomorphic relationship in Definition \autoref{defn:RegRevelationPrinciple} was indirectly characterized in \citet[pg.~291]{BodieRuffinoTreussard2008}(``[W]e apply [contingent claims analysis] CCA to the issue of career choices. The ability to change careers over one's working life is isomorphic with an option to exchange one asset for another."). This is similar to a \citet{Margrabe1978} option and lends itself to index executive option pricing. \hfill $\Box$
\end{rem}
\subsection{The value of the firm and related regulator call option}\label{subsec:ValueOfTheFirm}
In this section we derive the value of the firm. With few exceptions, the assumptions below are consistent with \citet[pg.~450]{Merton1974}. A significant departure is the asymmetric information induced by regulator signals that constitute a source of market friction that affects the firm's capital structure. See e.g., e.g. \citet{JensenMeckling1976}, \citet{Kwan2009}. Since the firm responds to the signal(s) it receives from the regulator, let $V(s(\theta))$ be the value of the firm\footnote{ The notation here is slightly more general than that for the definitions above.}. The regulator's potential contract or``managerial compensation" with the firm is functionally equivalent to an equity stake in the firm, $\alpha V,\;0<\alpha <1$, upon her career change. If $K\in X$ is the ``value" of the regulators current compensation package, then [s]he can exercise her [call] option at any time $t$. The erstwhile "abstract" topological lifting in Definition \autoref{defn:RegRevelationPrinciple} and Proposition \autoref{prop:RegulatorAbstractOption} shows that the stochastic choice function at time $t$ reduces to
\begin{equation}
   f(t;\omega,\theta)=(\alpha V(t,s(\omega,\theta))-K)^+
\end{equation}
More formally, that abstract [call] option contract takes the form
\begin{align}
   C(\tau^*;V,s(\theta)) &=\sup_{\tau\in\mathcal{T},\theta\in\Theta}(\alpha V(\tau,s(\theta))-K)^+\label{eq:RegAbstractOption}
\end{align}
where $\tau\in\mathcal{T}$ is a stopping time and $\tau^*$ is the optimal stopping time in the set $\mathcal{T}$ of all admissible \emph{early exercise times} for the regulator\footnote{\citet{LukasWelling2013} show how to derive optimal $\alpha$ and $\tau$ in a real option pricing model.}. The foregoing is summarized in the following
\begin{lem}[Regulator's career call option]\label{lem:RegulatorCareerCallOption}~\
   The value of the regulator's career call option contract $C(t,V,s(\theta))$ increases when embedded regulatory signals increases firm value $V(s(\theta))$. \hfill $\Box$
\end{lem}
Of course $V(s(\theta)|\mathcal{F}_t)$ is the value of the firm at time $t$ given the ``filtration" of information $\mathds{F}$, and it has to be discounted from some ``terminal date" $T$.  Specifically,  in the sequel firm value is assessed at the stopping time for the [former] regulator's early exercise date for her American style option. See e.g., \citet{Myneni1992} for a taxonomy of American option pricing. That is, firm value, and hence the career option, is computed in a \emph{sequential move setting} as follows:
\begin{enumerate}
   \item Regulator assesses her potential receptiveness and possible compensation premium from firm under extant circumstances.
   \item Given the assessment, the regulator decides to switch employment and sends.
   \item Regulator enters American style index executive stock option contract with firm upon career change.
   \item Regulator's early exercise date $\tau$ is used as terminal date to evaluate firm bankruptcy analysis in European style option setting
\end{enumerate}
The former regulator's executive stock option compensation scheme is influenced by her human capital beta. See e.g., \citet{HallMurphy2000}(managerial stock options); \citet{GracePhillips2008}(insurance commissioners human capital signaling) and \citet{Tahoun2013} (price of stocks held by Congress). More on point, given firm debt $D$, the firm's [signal dependent] equity is given by:
\begin{align}
   E(s(\theta)) &= V(s(\theta))-D\\
   \intertext{So the regulator's virtual call option is based on}
   C(\tau^*;V,s(\theta)) &=\sup_{\tau\in\mathcal{T},\theta\in\Theta}(\alpha V(\tau,s(\theta))-\alpha D)^+ \label{eq:RegulatorCallOptionOnFirmEquity}
\end{align}
This implies that the strike price in \eqref{eq:RegAbstractOption}, based on the value of the regulator's pre-career change compensation package, must be compared to the firm's debt firm debt $(\alpha D)$ which serves double duty as the strike price in \citet{Merton1974} firm valuation formula in \eqref{eq:RegulatorCallOptionOnFirmEquity}. This \citet{Margrabe1978} style option pricing problem was solved by \citet{Tressaud2007} and \citet{BodieRuffinoTreussard2008}. So we take it as given in the sequential move setting above.
\section[Valuation of regulator option with equity stake in the firm]{Valuation of regulator career option contract with equity stake in the firm}\label{sec:RegulatorOptionCapitalStructure}
We assume that the regulator's program consists of a prospective American style indexed option contract, with a regulated firm, which satisfies a linear growth condition. See e.g., \citet[pg.~174]{MusielaRutkowski2005}. Moreover, the index is based on stocks of a select peer group of competitor firms. See e.g., \citet[p.~323]{Jorgensen2002}; \citet[pg.~40]{JohnsonTian2000a}. We make the following
\begin{assumption}~\
   \begin{itemize}
       \item [1.] Index $I(t,\omega)$ represent common factor in firm performance \label{assum:CommonFacIndex}
       \item [2.] The firm's stock price is decomposable into impact of (1) a common factor, and (2) an idiosyncratic factor
   \end{itemize}
\end{assumption}
\subsection[Solving the regulator's program via American option contract]{Solving the regulator's program via American style indexed executive stock option contract}\label{subsec:RegulatorProgram}
Let $H(t,\omega)$ be the benchmark index observed by the firm and the regulator. We now show how the regulator use human capital beta to manipulate the benchmark. See e.g., \citet{DehaanKediaKohRajgopal2012} for empirical evidence on ``human capital hypothesis" in revolving door context. The regulator and firm enters an indexed [call] option contract $(S(t,\omega)-H(t,\omega))^+$ with asymmetric information in which the firm and regulator observes $H(t,\omega)$ but not the true benchmark $H(t,\omega;\;s(\theta))$ known to the regulator. Following, \citet[pp.~40-41]{JohnsonTian2000a}, we let the index $I(t,\omega)$ and stock price $S(t,\omega)$ follow geometric Brownian motion $(B)$ processes with correlation $\rho$ (after suppressing $(t,\omega)$) given by:
\begin{align}
   \frac{dS}{S} &= \mu_S dt + \sigma_SdB_S\\
   \frac{dI}{I} &= \mu_S dt + \sigma_IdB_I\\
   dB_S&\;dB_I = \rho dt
\end{align}
Here, we assume that dividend payout rate is zero. To characterize benchmark index dynamics, we apply the formulae in \citet[Prop.~2.3(i),~Thm.~4.3]{Jacka1991} for an American put option by changing the signs and probability statement in Jacka's equation (4.3)\footnote{Jacka's $K(x,t)$ is our $C(t,S;\theta)$. \citet[pp.~551,~553]{Kim1990} solves the American call option problem, with early exercise premium, using iterative methods.} accordingly. In that case the regulator's call option is decomposed as:
\begin{equation}
   C(t,S;r) = H(0)-S(0) + rH(0)\int^t_0e^{-ru}\text{Pr}\{S(u,\omega)>H(u,\omega)\}du\label{eq:DecomposeRegOption}
\end{equation}
where $H(u,\omega)$ is a \emph{de facto} time varying boundary which gives rise to the \emph{stoping time}
\begin{equation}
   \tau^H(\omega) = \{t\ge 0; S(t,\omega) = H(t,\omega)\}\label{eq:RegStopTime}
\end{equation}
From \eqref{eq:DecomposeRegOption} and \eqref{eq:RegStopTime} the \emph{regulator's early exercise premium } (REEP) is given by the quantity:
\begin{equation}
   \text{REEP}(\tau^H)=rH(0)\int^{\tau^H}_0e^{-ru}\text{Pr}\{S(t,\omega)\ge H(t,\omega)\}du\label{eq:REEP}
\end{equation}
When $H(t,\omega)$ decreases, the probability $\text{Pr}\{\cdot\}$ increases and the value of the early exercise premium REEP$(\tau^H)$ increases as well.

To establish a nexus between regulator signal and stock price volatility we turn to the index option pricing models in \citet[pg.~40]{JohnsonTian2000a} and \citet[pg.~328]{Jorgensen2002}. There, the benchmark stock price $(H)$ is obtained from lognormality assumptions, see \citet[pg.~533]{Hull2006}, which, in the context of our model gives us:
\begin{align}
   H(t)=E^P[S(t,\omega)] &= S(0)\Bigg(\frac{I(t)}{I(0)}\Bigg)^{\beta^{vw}} e^{\eta t}\label{eq:ObsBenchmark}
\end{align}
\begin{align}
   \eta &= (r+\frac{1}{2}\rho\sigma_S\sigma_I)(1-\beta^{vw})\label{eq:ExcessRet}
\end{align}
Here, $\eta$ is excess returns; $r$ is a riskless interest rate, and $\sigma_S,\;\sigma_I$ are volatilities; $\beta^{vw}=\rho(\frac{\sigma_S}{\sigma_I})$ is a measure of systemic risk with the [value weighted] index in a pseudo-CAPM framework. See e.g., \citet[pp.~41-42]{JohnsonTian2000a} for further details on formulae.
\subsubsection{The regulator's human capital beta}\label{subsubsec:regHumanCapBeta}
The interplay between firm capital structure and employee compensation is fairly established in the literature. For example, \citet[pp.~907-908]{BerkStantonZechner2010} introduced a model in which firms with higher leverage pay higher wages. And \citet{ChemmanurChengZhang2010} provides empirical evidence that the compensation of newly hired CEOs is positively impacted by firm leverage. In this section we employ \citet[pg.~15,~eq(22)]{JagannathanWang1996} human capital beta pricing model to establish a nexus between firm leverage and compensation. Those authors decomposed CAPM beta into a value weighted and labor (human capital) component as follows:
\begin{equation}
   \beta = b_{\footnotesize{vw}}\beta^{\footnotesize{vw}} + b_{\footnotesize{labor}}\beta^{\footnotesize{labor}}
\end{equation}
If we take $b_{\footnotesize{vw}}$ as \emph{numeraire} we can rewrite the equation as
\begin{equation}
  \beta^\prime = \beta^{\footnotesize{vw}} + b^\prime_{\footnotesize{labor}}\beta^{\footnotesize{labor}}\label{eq:NumeraireBeta}
\end{equation}
However, the $\beta$ in \eqref{eq:ObsBenchmark} is based on the [value weighted] benchmark. The absence of human capital beta in that benchmark implies that equation \eqref{eq:ObsBenchmark} is misspecified. Substitution of $\beta^{\footnotesize{vw}}$ from \eqref{eq:NumeraireBeta} in \eqref{eq:ExcessRet} gives us
\begin{equation}
   \widehat{\eta}=\eta +  b^\prime_{\footnotesize{labor}}\beta^{\footnotesize{labor}}\label{eq:ModExcessRet}
\end{equation}
where $\eta=1-\beta^\prime$. Thus, the true benchmark (unobserved by the firm)with human capital signals is given by
\begin{align}
   H(t,\beta^{\footnotesize{labor}}) &= S(0)\Bigg(\frac{I(t)}{I(0)}\Bigg)^{\beta^\prime}\exp({\widehat{\eta}t})\label{eq:UnObsBenchmark}
\end{align}
\begin{align}
   &=S(0)\Bigg(\frac{I(t)}{I(0)}\Bigg)^{\beta}\exp({\eta t})\Bigg(\frac{I(t)}{I(0)}\Bigg)^{b^\prime_{\footnotesize{labor}}\beta^{\footnotesize{labor}}}\exp({b^\prime_{\footnotesize{labor}}\beta^{\footnotesize{labor}}t})\\
   &= H(t)\Bigg(\frac{I(t)}{I(0)}\Bigg)^{b^\prime_{\footnotesize{labor}}\beta^{\footnotesize{labor}}}\exp({b^\prime_{\footnotesize{labor}}\beta^{\footnotesize{labor}}t})\label{eq:MultBenchmark}\\
               \Rightarrow H(t) &= H(t,\beta^{\footnotesize{labor}})\Bigg(\frac{I(t)}{I(0)}\Bigg)^{-b^\prime_{\footnotesize{labor}}\beta^{\footnotesize{labor}}}\exp(-{b^\prime_{\footnotesize{labor}}\beta^{\footnotesize{labor}}t})\label{eq:DisBenchmark}
\end{align}
The \emph{implicit} lognormal relationship in \eqref{eq:MultBenchmark}, see e.g., \citet[pp.~274-275]{Hull2006},  suggests that after cancelling out the common time variable $t$
\begin{equation}
   \sigma^2_H\Big(\beta^{\footnotesize{labor}}\Big)=\sigma^2_H+\Big(b^\prime_{\footnotesize{labor}}\beta^{\footnotesize{labor}}\Big)^2\sigma^2_I \;>\; \sigma^2_H\label{eq:LogNormalSigmaIneq}
\end{equation}
Thus, by increasing human capital signal $\beta^{\footnotesize{labor}}$ the regulator is able to (1) increase true benchmark volatility, and (2) lower the contractual benchmark in \eqref{eq:DisBenchmark}\footnote{This result is invariant to regulator type. For example, following the conditional beta paradigm in \citet[pp.~429-430]{FersonSchadt1996}, rewrite \eqref{eq:NumeraireBeta} as $\beta=(b_0+b_1Z)\beta^{\footnotesize{labor}}+b_{\footnotesize{vw}}\beta^{\footnotesize{vw}}$ where for some threshold $\beta^\star$, $Z=1$ for high beta type $\beta^{\footnotesize{labor}}>\beta^\star$ and $Z=0$ for low beta type $\beta^{\footnotesize{labor}}\leq \beta^\star$. Thus, $\beta^{\footnotesize{labor}}\mapsto b_0\beta^{\footnotesize{labor}}$ for low or $\beta^{\footnotesize{labor}}\mapsto (b_0+b_1)\beta^{\footnotesize{labor}}$ for high types. For each type, $\sigma^2_H\Big(\beta^{\footnotesize{labor}}\Big)\;>\; \sigma^2_H$.}. Whereupon, [s]he is able to increase the early exercise premium (REEP) in \eqref{eq:REEP} and consequently the value of her call option. This foregoing analysis is summarized the following
\begin{lem}[Regulator's induced volatility]\label{lem:RegInducedVolatility}~\\
   Let $\theta=\beta^{\footnotesize{labor}}$; $H(t)$ be the benchmark in \eqref{eq:ObsBenchmark} observed by firm and regulator; $H(t,\theta)$ be the benchmark in \eqref{eq:UnObsBenchmark} known only be the regulator; and $S(t,\omega,\theta)=H(t,\theta) + \epsilon(t,\omega)$ be the firm's stock price, where $\epsilon\sim\;iid(0,\sigma^2_\epsilon)$. Let $\sigma_S$ be stock price volatility without regulator signalling, and $\sigma_S(\theta)$ be stock price with regulator signalling. Then $\sigma_S(\theta)>\sigma_S$. \hfill$\Box$
\end{lem}
\begin{proof}
   See e.g., \autoref{apx:ProofOfRegInducVolat}.
\end{proof}
\begin{rem}\label{rem:KanePenaltyAdjBeta}
   If we include a Kane-type penalty function $\Delta>0$ in the beta pricing relationship in \eqref{eq:NumeraireBeta}, then in the event of regulatory failure it becomes  $\beta^\prime_\Delta = \beta^{\footnotesize{vw}} + (b^\prime_{\footnotesize{labor}}-\Delta)\beta^{\footnotesize{labor}}$. In which case volatility in \eqref{eq:LogNormalSigmaIneq} is reduced and the value of the option is reduced as well.
\end{rem}
\subsection{Embedding regulator signals in firm equity}\label{subsec:RegSignalInEquity}
In the last section we set $\theta=\beta^{\footnotesize{labor}}$ to illustrate how human capital beta influences the career call option contract. For notational convenience, we use $\theta$ in the sequel. In the context of \eqref{eq:RegulatorCallOptionOnFirmEquity}, we apply \citet[pp.~453-454]{Merton1974} formula for the option value of the equity in a firm with value $V$ and background driving Brownian motion, and debt $D$. Specifically,
\begin{subequations}\label{eqgrp:FirmOptionValue}
  \begin{align}
     f(V,\tau) &= V\Phi(x_1)-D\exp(-r\tau)\Phi(x_2)\label{eqgrp:FirmOptionValueEqu}\\
     x_1 &= \frac{\log\tfrac{V}{D}+(r+\tfrac{\sigma^2}{2})\tau}{\sigma\sqrt{\tau}}\\
     x_2 &= x_1-\sigma\sqrt{\tau}
  \end{align}
\end{subequations}
where $f(V,\tau)$ is option value of firm equity, $r$ is a risk free rate, $\tau=T-t$ for exercise date $T$ for a European call option priced at time $t$; and $\Phi$ is the cumulative normal. From the outset we note that
\begin{align}
   \alpha f(V,\tau) &= \alpha V\Phi(x_1)-\alpha D\exp(-r\tau)\Phi(x_2)\\
   \intertext{The right hand side of that equation is isomorphic to a call option}
   C(\cdot) &=(\alpha V-\alpha D)^+\\
   \intertext{So the regulator's valuation rests on comparing her current compensation package $K$ to the psuedo-strike price $\alpha D$ that depends on the firm's debt. That is, the exercise of her career change option rests on the evaluation}
   (\alpha V-K)^+&\gtreqqless(\alpha V-\alpha D)^+
\end{align}
We assume that the regulator knows $V, K$ and $D$ but she does not know fraction $\alpha$ on the early exercise date. Thus, if she believes that $K<\alpha D$ so that $(\alpha V-K)^+ > (\alpha V-\alpha D)^+$, then her valuation of the firm is greater than its actual value. So she would be inclined to exercise a Margrabe option to ``exchange" the subject cash flows by quitting accordingly, and joining the firm\footnote{This problem was solved by \citet{Tressaud2007} and \citet{BodieRuffinoTreussard2008}, and is not dispositive of our main results.}. What is more, the regulator now has incentive to design a mechanism that increases firm leverage. So she looks more favorably at levered firms because a highly levered firm increases the value of her career change option. This prediction is consistent with the empirical evidence which finds that privately controlled firms use leverage strategically to obtain better regulatory outcomes. See e.g., s.g. \citet{BortolottiCambiniRondiSpiegel2011}. More important, it explains the finding in \citet[pp.~217-218]{DalBo2006} which finds that regulators increase their support for industry interest during their last year in office. Therefore, in the sequel we assume that $K < \alpha D$.

\subsection{Regulator career option pricing with signalling}\label{subsec:RegulatorOptionPriceWithSignals}
Because $\theta$ is unobservable, its presence in firm value, i.e. stock price, implies that the latter may follow a hidden Markov process. \citet{Guo2001} and \citet{Shepp2002} introduced an information based stock price model on the canonical probability measure space $(\Omega,\mathcal{F},\mathrm{\mathrm{F}},P)$ with filtration $\mathrm{F}$,  which, in the context of Proposition \autoref{prop:RegulatorAbstractOption}, is parameterized as
\begin{align}
   \frac{dV(t,\omega)}{V(t,\omega)} &=\mu(s(t,\omega,\theta))dt+\sigma(s(t,\omega,\theta))dW(t,\omega)
\end{align}
where $s(t,\omega,\theta)$ is itself an unobservable [hidden] Markov process. However, to simplify matters we will assume that $s(t,\omega,\theta)=s(\theta)$. That is, the signal is a deterministic control\footnote{See e.g., \citet[pg.~237]{Oksendal2003}}. In that case, the problem is reduced to one of an American style option where the regulator embeds a deterministic signal in the value of the firm, and she can exercise her option to join the hitherto regulated firm at any time. Thus, firm value dynamics is given by
\begin{align}
    \frac{dV(t,\omega)}{V(t,\omega)} &=\mu(s(\theta))dt+\sigma(s(\theta))dW(t,\omega)\label{eq:FirmValuDynamics}
\end{align}
More on point, \citet[pg.~1373]{Shepp2002} and \citet{Guo2001} introduce evidence that information asymmetry is dispositive of volatility dynamics, and the activity of insiders leads to increased volatility. That artifact of their model is consistent with \citet{BoyleJhaKennedyTian2011} who find that high volatility firms should use more stock options to compensate their executives. See e.g., also, \citet[\S5.4,~pg.~104]{HaugTaleb2011} for information dependent volatility in option pricing. In our model, those phenomena are represented by Lemma \autoref{lem:RegInducedVolatility} on page \pageref{lem:RegInducedVolatility}. Thus, in \eqref{eq:FirmValuDynamics} we have
\begin{equation}
   \sigma(s(\theta)) > \sigma \label{eq:FirmSignalVolGtNonSigVol}
\end{equation}
Using the formulae in \eqref{eqgrp:FirmOptionValue}, we find that
\begin{subequations}\label{eqgrp:FirmOptionValueWitSignal}
   \begin{align}
      f(V(s(\theta)),\tau) &= V(s(\theta))\Phi(x_1)-D(s(\theta))\exp(-r\tau)\Phi(x_2))\label{eqgrp:FirmOptionValueEquWitSignal}\\
      x_1(s(\theta)) &= \frac{\log\tfrac{V(s(\theta))}{D(s(\theta))}+(r+\tfrac{\sigma(s(\theta))^2}{2})\tau}{\sigma(s(\theta))\sqrt{\tau}}\\
      x_2(s(\theta)) &= x_1(s(\theta))-\sigma(s(\theta))\sqrt{\tau}
  \end{align}
\end{subequations}

In the sequel we write $x_1(\theta)$ and $x_2(\theta)$ instead of $x_1(s(\theta))$ and $x_2(s(\theta))$. To ascertain the impact of regulator signals on firm value, we compare the value of the option on the firm with embedded regulator signal to the value without signal. To do that we turn to option Greeks by examining the vega $(\vartheta)$ of the option on the firm\footnote{See e.g., \citet[pg.~361]{Hull2006} for derivation of formula for vega.}. That is, for the mean zero normal distribution function $\phi$, let $\vartheta(0)$ and $\vartheta(s(\theta))$ be the vega without and with regulator signals respectively. And define
\begin{align}
   \vartheta(0) =\frac{\partial f(V,\tau)}{\partial\sigma}&=V\phi(x_1)\sqrt{\tau}\\
   \vartheta(s(\theta))=\frac{\partial f(V(s(\theta),\tau)}{\partial\sigma(s(\theta))}&=V(s(\theta))\phi(x_1(\theta))\sqrt{\tau}\\
   \intertext{According to Lemma \autoref{lem:RegInducedVolatility}}
   \vartheta(s(\theta))&>\vartheta(0)\\
   \Rightarrow V(s(\theta))\phi(x_1(\theta))\sqrt{\tau}&>V\phi(x_1)\sqrt{\tau}\Rightarrow\frac{V(s(\theta))}{V}>\frac{\phi(x_1)}{\phi(x_1(\theta))}\\
   \intertext{According to Lemma \autoref{lem:RegulatorCareerCallOption} a more levered firm with embedded regulator signals enhances regulator career option. Thus, we have}
   \frac{\phi(x_1)}{\phi(x_1(\theta))}>1\Rightarrow & \exp\Biggl(-\Biggl(\frac{x_1^2}{2}-\frac{x_1(\theta)^2}{2}\Biggr)\Biggr)>1\\
   \Rightarrow &-x_1(\theta)<-x_1 \;\;\text{and}\;\; x_1(\theta)>x_1
\end{align}
\citet[pp.~405-406]{GroppVesalaVulpes2006} defined $-x_1$ as the ``distance to default'' and found that it was a predictor of bank stability and or bankruptcy. In fact, they report that ``the banks that were downgraded  * * * had a significantly higher mean value of $[x_1]$ than those that did not" in the 6-18 months prior to downgrade. See e.g., also \citet{RavivCiamarra2013}. The analysis below provides a theoretical explanation.
\section[Bankruptcy, regulatory signals, firm vega and leverage]{Bankruptcy with regulatory signals in firm vega and leverage}\label{sec:WarningSignalsVegaAvgLev}
The foregoing arguments show that the vega for the option value of the firm is dispositive of regulator incentive to leverage regulated firms. See e.g.,, \citet{ArmstrongVashishtha2012}. According to Kolmogorov Inequality, see e.g. \citet[Thm.~8.3.1,~pg.~249]{AthreyaLahiri2006}, and \eqref{eq:FirmSignalVolGtNonSigVol}, if $\varphi$ is the proportion of total tail probability in the lower tail, then for some threshold value $\lambda$ and assuming maximal probability thresholds we have the nonparametric result
\begin{align}
    &\text{Pr}\Bigl(\max_\theta\{x_1(\theta)\}<-\lambda\Bigr)=\varphi\frac{\text{Var}(x_1(\theta))}{\lambda^2}\\
    &=\varphi\frac{\widetilde{\sigma}^2_{x_1}(\theta)}{\lambda^2}>\varphi\frac{\text{Var}(x_1)}{\lambda^2}=\varphi\frac{\widetilde{\sigma}^2_{x_1}}{\lambda^2}\label{eq:RegValueAtRisk}
\end{align}
where $\widetilde{\sigma}^2_{x_1}(\theta)$ and $\widetilde{\sigma}^2_{x_1}$ are the variances of $x_1(\theta)$ and $x_1$ respectively. The result in \eqref{eq:RegValueAtRisk} is important because it says that the ``large deviation probability"\footnote{See e.g., \citet{Varadhan2008}} for firm value at risk included in $x_1(\omega_b,\theta)$, i.e. tail risk in bad states $\omega_b\in\Omega$, is greater than the corresponding probability if there was no regulatory signals or ``regulatory capture'' embedded in capital structure\footnote{By the same token, if regulator signals induce less volatility, i.e. $\widetilde{\sigma}^2_{x_1}(\theta) < \widetilde{\sigma}^2_{x_1}$ then the firm is less likely to be bankrupt but stock option based compensation will be reduced. See e.g., \citet{ArmstrongVashishtha2012}.}. To see this rewrite $V(s(\theta))$  in \eqref{eqgrp:FirmOptionValueWitSignal} as $V(\theta)$ so that after algebraic reduction we have
\begin{equation}
   \lim_{\sigma(\theta)\rightarrow\infty}\text{Pr}\Bigl\{\max_\theta V(\theta)<\max_\theta D(\theta)\exp\Bigl(-\lambda\sigma(\theta)\sqrt{\tau}-\Bigl(r+\tfrac{\sigma^2(\theta)}{2}\Bigr)\tau\Bigr)\Bigr\}>\varphi\frac{\widetilde{\sigma}^2_{x_1}}{\lambda^2}
\end{equation}
\begin{align}
   \Rightarrow\text{Pr}\{\max_\theta V(\theta)<0\}&>\varphi\frac{1}{\lambda^2\sigma^2\tau}\text{Var}\Bigl(\ln\frac{V}{D}\Bigr)\label{eq:BankruptcyCondition}
\end{align}
where $\widetilde{\sigma}^2_{x_1}=\tfrac{1}{\sigma^2\tau}\text{Var}(\ln\frac{V}{D})$. \autoref{eq:BankruptcyCondition} is equivalent to a bankruptcy condition. It says that if regulator signals induce too much volatility, i.e. uncertainty, in firm value, then the probability that the value of the firm is negative, i.e. it goes bankrupt, is greater than a given threshold without regulator signals. This result characterizes Lehman Brothers bankruptcy where regulators signalled that they were not going to provide liquidity for a highly levered bank. See e.g., \citet{AdrianShin2010} and \citet{TaylorTzengWiddick2012}. The Kane penalty function in Remark \autoref{rem:KanePenaltyAdjBeta} mitigates this risk.

Thus we have proven the following
\begin{prop}\label{prop:RegCaptureValueAtRisk}[Bankruptcy condition]~\label{prop:BankruptcyCondition}
    Assume that regulators embed regulatory signals $\theta$ in firm capital structure. Then firm value-at-risk $x_1(\omega_b,\theta)$ in bad states of nature $(\omega_b)$ is greater than $x_1(\omega_b)$--its value in non-revolving door regimes--and the probability of bankruptcy is higher than it would be in the absence of regulator signals. \hfill $\Box$
\end{prop}
\begin{rem}
   This proposition should not be interpreted to mean that in a world with no regulation, firm value-at-risk will be lower than with regulation. The proposition is based on an assumption of asymmetric information between firm and regulator. Thus, more transparency and even imposition of a Kane penalty mitigates against our bankruptcy result. \hfill $\Box$
\end{rem}

\section{Conclusion}\label{sec:Conclusion}
This paper derives a firm bankruptcy risk result based on a regulator's career concerns in a revolving door setting. We argue that before career-switch, regulators embed their human capital [beta] in a target firm's capital structure. We show how regulators use their human capital beta to affect the early exercise premium in their subsequent indexed executive stock option. Further, we show how that behaviour increases a firm's bankruptcy risk, and estimation of its value at risk (VaR). In the United States, conflict of interest and anti-lobbying laws are in place to mitigate against some of the predictions of our theoretical model. See e.g., 18 U.S.C. \S207, 21 U.S.C. \S1602(3). To the best of out knowledge, the results in this paper are new to the literature on corporate finance and mechanism design. However, the overarching theme is how the vagaries of regulation contributes to financial instability. Of course, we have one caveat.  Like public servants everywhere, regulators approach their careers with varying degrees of commitment to public service.  One set of these will be open to and will entertain alternatives to public service.  Our model is relevant to this subset.  However, we would expect that there are others whose passion is fed by a public trust, and as such are not expected to do other than adhere to their expected charge of bringing maximum consumer surplus benefits from their charge as regulators.
\singlespace
\appendix{
\begin{center}
\sc{Appendix}
\end{center}
}\label{apx:Appendix}
\section{Proofs}\label{apx:Proofs}
\subsection{Proof of incentive compatibility constraint \autoref{defn:RegIncentiveConstraint}}\label{apx:ProofOfICconstraint}
   We assume that the regulator enters an indexed option contract $C(S(t,\omega), H(t,\omega);\;\theta)=(S(t,\omega)-H(t,\omega))^+$ with stock price $S(t,\omega)$, and time varying benchmark $H$ related to her type $\theta$ such that [suppressing $\omega$] $H(t)=H(t,\theta)e^{-\theta t}$. Let $\widehat{\theta}=\theta + \epsilon$ where $\epsilon\sim (0,\sigma^2_\epsilon)$. Assume that $H$ is twice differentiable in $\theta$. Thus,
   \begin{align}
      \widehat{H}(t,\theta) &= H(t,\widehat{\theta})e^{-\widehat{\theta} t}\\
      &\approx \Bigg(H(t,\theta) + H_{\widehat{\theta}}(t,\theta)\epsilon + H_{\widehat{\theta}\widehat{\theta}}(t,\theta)\frac{1}{2}\epsilon^2\Bigg)e^{-\theta t}e^{-\epsilon t}\\
      &\Rightarrow E^P[\widehat{H}(t)]\approx E^P\Bigg[\Big(H(t,\theta)+H_{\widehat{\theta}}(t,\theta)\epsilon + H_{\widehat{\theta}\widehat{\theta}}(t,\theta)\frac{1}{2}\epsilon^2\Big)\Big(1-\epsilon t +\frac{1}{2}\epsilon^2 t^2\Big)e^{-\theta t}\Bigg]\\
      &= H(t) + \Big(\frac{1}{2}H_{\widehat{\theta}\widehat{\theta}}(t,\theta)-H_{\widehat{\theta}}(t,\theta)t+\frac{1}{2}H(t,\theta)t^2\Big)\sigma^2_\epsilon e^{-\theta t}\\
      &\Rightarrow E^P[\widehat{H}(t)] \ge H(t)
   \end{align}
   So by not revealing her true type our regulator will set a higher performance benchmark for herself! And the value of $C$ will be reduced. Therefore, it is in her interest to reveal her true type in order to lower the benchmark for performance. \hfill$\Box$
\subsection{Proof of Lemma \autoref{lem:RegInducedVolatility}}\label{apx:ProofOfRegInducVolat}
   \begin{proof}
      When there is no signal, we have $\sigma^2_S=\sigma^2_H+\sigma^2_\epsilon$. With signals, according to \eqref{eq:LogNormalSigmaIneq}, we have $\sigma^2_S(\theta)=\sigma^2_H+\Big(b^\prime_{\footnotesize{labor}}\beta^{\footnotesize{labor}}\Big)^2\sigma^2_I + \sigma^2_\epsilon > \sigma^2_S=\sigma^2_H+\sigma^2_\epsilon$.
   \end{proof}
\bibliographystyle{chicago}
\addcontentsline{toc}{section}{References} 

\end{document}